\documentclass[11pt]{article}
\usepackage{graphicx}
\usepackage{float}
\usepackage[utf8]{inputenc}
\usepackage[english]{babel}
\usepackage{amsmath,amsthm}
\usepackage{amssymb}
\usepackage{amsfonts}
\usepackage{CJK}

\usepackage{hyperref}

\hypersetup{
bookmarks=true,
colorlinks=true,
linkcolor=blue,
urlcolor=blue,
citecolor=blue,
pdftex,
linktocpage=true, 
hyperindex=true
}

\newenvironment{Proof}{{\noindent\it Proof:  }}{\hfill\rule{2mm}{2mm}}
\newenvironment{proofof}[1]{{\noindent\it Proof of #1:  }}{}

\newtheorem{theorem}{Theorem}[section]
\newtheorem{definition}[theorem]{Definition}

\newtheorem{lemma}[theorem]{Lemma}
\newtheorem{algo}{Algorithm}

\newcommand{\N}{\ensuremath{\mathbb N}}

\newcommand{\poly}{\ensuremath{\mathsf{poly}}}

\begin{document}
\title{Synchronization Strings: Efficient and Fast Deterministic Constructions over Small Alphabets}
\author{Kuan Cheng\thanks{kcheng17@jhu.edu. Department of Computer Science, Johns Hopkins University.   Supported by NSF Grant CCF-1617713.}
 \and Xin Li\thanks{lixints@cs.jhu.edu. Department of Computer Science, Johns Hopkins University.  Supported by NSF Grant CCF-1617713.}
  \and Ke Wu\thanks{AshleyMo@jhu.edu. Department of Computer Science, Johns Hopkins University.}
  }

\date{}
\maketitle{}

\begin{abstract}
Synchronization strings are recently introduced by Haeupler and Shahrasbi \cite{haeupler2017synchronization} in the study of codes for correcting insertion and deletion errors (insdel codes).\ A synchronization string is an encoding of the indices of the symbols in a string, and together with an appropriate decoding algorithm it can transform insertion and deletion errors into standard symbol erasures and corruptions. This reduces the problem of constructing insdel codes to the problem of constructing standard error correcting codes, which is much better understood. Besides this, synchronization strings are also useful in other applications such as synchronization sequences and interactive coding schemes.

Amazingly, Haeupler and Shahrasbi \cite{haeupler2017synchronization} showed that for any error parameter $\varepsilon>0$, synchronization strings of arbitrary length exist over an alphabet whose size depends only on $\varepsilon$. Specifically, \cite{haeupler2017synchronization} obtained an alphabet size of $O(\varepsilon^{-4})$, as well as a randomized construction that runs in expected time $O(n^5)$. However, it remains an interesting question to find deterministic and more efficient constructions.

In this paper, we improve the construction in \cite{haeupler2017synchronization} in three aspects: we achieve a smaller alphabet size, a deterministic construction, and a faster algorithm. Along the way we introduce a new combinatorial object, and establish a new connection between synchronization strings and insdel codes --- such codes can be used in a simple way to construct synchronization strings. This new connection complements the connection found in \cite{haeupler2017synchronization}, and may be of independent interest. In an independent work \cite{HS17c}, Haeupler and Shahrasbi also give deterministic constructions of synchronization strings over arbitrary length (or even infinite length). Their constructions can achieve linear construction time, but have alphabet size $\varepsilon^{-O(1)}$, which may be larger than ours.
\end{abstract}

\newpage

\section{Introduction}
The general and most important goal of coding theory is to ensure the transmission of messages reliably in the presence of noise or adversarial error. Starting from the pioneering works of Shannon, Hamming and many others, coding theory has evolved into an extensively studied field, with applications found in various areas in computer science. Regarding the general goal of correcting errors, we now have an almost completely understanding of how to deal with symbol erasures and corruptions. On the other hand, the knowledge of codes for timing errors such as insertion and deletion errors, has lagged far behind despite also being studied intensively since the 1960s. In practice, this is one of the main reasons why communication systems require a lot of effort and resources to maintain synchronization strictly.

Intuitively, one major difficulty in designing codes for insertion and deletion errors is that in the received codeword, the positions of the symbols may have changed. This is in contrast to standard symbol erasures and corruptions, where the positions of the symbols always stay the same. Thus, many of the known techniques in designing codes for standard symbol erasures and corruptions, cannot be directly applied to the case of insertion and deletion errors. Naturally, if one can find a way to bridge this gap and transform insertion and deletion errors into symbol erasures and corruptions, this will make our life much easier. 

In a recent work \cite{haeupler2017synchronization},  Haeupler and Shahrasbi introduced a combinatorial object called \emph{synchronization strings} to achieve exactly this goal. Informally, a synchronization string of length $n$ is an encoding of the indices of the $n$ positions into one string over some alphabet $\Sigma$, such that despite some insertion and deletion errors, one can still recover the correct indices of many symbols. Once we know the correct indices of these symbols, a standard error correcting code can then be used to recover the original message. This then gives a code for insertion and deletion errors, which is the combination of a standard error correcting code and a synchronization string.

The simplest example of a synchronization string is just to record the index of each symbol, i.e, the string $1, 2, \cdots, n$. It can be easily checked that even if $(1-\varepsilon)$ fraction of these indices are deleted, one can still correctly recover the remaining $\varepsilon$ fraction. However, this  synchronization string uses an alphabet whose size grows with the length of the string. The main contribution of \cite{haeupler2017synchronization} is to show that under a slight relaxation, there exist synchronization strings of arbitrary length $n$ over an alphabet with fixed size. Furthermore, \cite{haeupler2017synchronization} showed a very efficient (in fact, streaming) way to recover the indices of many symbols correctly from a synchronization string after insertion and deletion errors. Together this gives a code that for any $\delta \in (0,1)$ and $\varepsilon>0$, can correct $\delta$ fraction of insertion and deletion errors with rate $1-\delta-\varepsilon$. 

Besides this, synchronization strings have found a variety of applications, such as in  synchronization sequences \cite{mercier2010survey}, interactive coding schemes \cite{gelles2015coding, gelles2015capacity, ghaffari2014optimal1, ghaffari2014optimal2, haeupler2014interactive, kol2013interactive, haeupler2017synsimucode}, and edit distance tree codes \cite{braverman2017coding}. Furthermore, because of the nice properties of synchronization strings, it is plausible that they will find other applications in the future. However, despite the usefulness of such objects, it remains an interesting open problem to find deterministic and more efficient constructions of synchronization strings, as in \cite{haeupler2017synchronization} the authors only give a randomized construction.

To discuss the work of \cite{haeupler2017synchronization} and synchronization strings in more details, we first need the following formal definition of a synchronization string.

\begin{definition} \cite{haeupler2017synchronization} \label{def:sc} ($\varepsilon$-synchronization string)  For some alphabet $\Sigma$, a string $S \in \Sigma^n$ is an $\varepsilon$-synchronization string if $\forall 1\leq i < j <k\leq n$, we have that $ED(S[i,j],S[j+1,k])>(1-\varepsilon)(k-i)$ where $ED( , )$ stands for the edit distance of two strings, and $S[i,j]$ means the continuous subsequence of $S$ from $i$th position to $j$th position, both ends included.
\end{definition}

In \cite{haeupler2017synchronization}, Haeupler and Shahrasbi showed that for any $n \in \N$, $\varepsilon$-synchronization strings with length $n$ exist over an alphabet of size $O(\varepsilon^{-4})$. They further gave a randomized algorithm to construct such strings with expected running time $O(n^5)$. In this paper, we improve their construction in the following three aspects: we achieve a smaller alphabet size, a deterministic construction instead of a randomized construction, and a faster algorithm to construct synchronization strings.

\subsection{Our result}
Our first result shows the existence of $\varepsilon$-synchronization strings over a smaller alphabet, and in addition there is a randomized algorithm to compute such strings in expected polynomial time:

\begin{theorem}
For any $\varepsilon\in(0,1)$ and any $n \in \N$, there exists an $\varepsilon$-synchronization string $S$ of length $n$ over alphabet $\Sigma$ with $|\Sigma| =O(\varepsilon^{-2})$. In addition, there exists a randomized algorithm that can construct such a string in expected time $O(n^5\log n)$.
\end{theorem}

Next, we give deterministic polynomial time constructions for $\varepsilon$-synchronization strings when $\varepsilon$ is any constant, albeit with a slightly larger alphabet size.

\begin{theorem}
For any constant $\varepsilon\in(0,1)$ and any $n \in \N$, an $\varepsilon$-synchronization string $S$ of length $n$ over alphabet $\Sigma$ with $|\Sigma| = O(\varepsilon^{-3})$ can be constructed deterministically in time $\poly(n)$.
\end{theorem}

Our next theorem can handle smaller $\varepsilon$, i.e., $\varepsilon=o(1)$. We also significantly improve the time to construct an $\varepsilon$-synchronization string. In fact, we achieve near linear construction time, at the price of increasing the alphabet size again by a factor of $O(1/\varepsilon)$.

\begin{theorem}
There exists a constant $C>1$ such that for any $n \in \N$ and any $\varepsilon \geq \frac{C (\log \log n)^2}{\log n}$, an $\varepsilon$-synchronization string $S$ of length $n$ over alphabet $\Sigma$ with $|\Sigma| = O(\varepsilon^{-4})$ can be constructed deterministically in time $O(n\cdot  (\log \log n)^2)$.
\end{theorem}

In an independent work \cite{HS17c}, Haeupler and Shahrasbi also give deterministic constructions of synchronization strings over arbitrary length (or even infinite length). Their constructions can achieve linear construction time, but have alphabet size $\varepsilon^{-O(1)}$, which may be larger than ours.

\subsection{Our techniques}
To prove the existence of $\varepsilon$-synchronization strings over an alphabet of size $O(\varepsilon^{-2})$, we modify the existence proof in \cite{haeupler2017synchronization}. Note that for a string $S$ to be an $\varepsilon$-synchronization string, we need to make sure that every interval $S[i, k]$ of $S$ satisfies the property in Definition~\ref{def:sc}. In \cite{haeupler2017synchronization}, the authors proved the existence of $\varepsilon$-synchronization strings over an alphabet of size $O(\varepsilon^{-4})$, by dealing with two cases separately. The first case is where the interval of interest is short, i.e., with length at most $\varepsilon^{-2}$. In this case the authors use an alphabet $\Sigma_1$ of size $m=\varepsilon^{-2}$ and mark every position $i$ with the symbol $i$ mod $m$. This ensures that every small interval has distinct symbols. The second case is where the interval of interest is long, i.e, with length at least $\varepsilon^{-2}$. The authors handle this case by uniformly randomly choosing a symbol for each position from another alphabet $\Sigma_2$ of size $O(\varepsilon^{-2})$. They then use the General Lov\'asz Local Lemma to show that with positive probability, there exists a choice of the symbols that ensures every large interval satisfies the property in Definition~\ref{def:sc}. The final synchronization string is then obtained by combining these two cases, i.e., $\Sigma =\Sigma_1 \times \Sigma_2$ and this results in an alphabet of size $O(\varepsilon^{-4})$. In \cite{haeupler2017synchronization}, the authors conjectured that by using a non-uniform sample space in the General Lov\'asz Local Lemma, one could potentially avoid the use of $\Sigma_1$ and thus reduce the alphabet size to $O(\varepsilon^{-2})$.

Here we confirm their conjecture and indeed present a proof of the existence of $\varepsilon$-synchronization strings over an alphabet of size $O(\varepsilon^{-2})$, based on using non-uniform sample space in the General Lov\'asz Local Lemma. Specifically, fix an alphabet $\Sigma$ and a length $t=\varepsilon^{-2}$, we pick the symbols for each position in the string $S$ as follows: pick the first symbol uniformly randomly from $\Sigma$. Then, for the $i$'th position, we uniformly randomly pick a symbol from $\Sigma$, conditioned on that this symbol is distinct from the previous $t-1$ or $i-1$ symbols, whichever is smaller.

Note that this way of choosing symbols also guarantees that every interval of length at most $\varepsilon^{-2}$ has distinct symbols. We then use a similar proof as in \cite{haeupler2017synchronization} to prove that again, with positive probability, there exists a choice of the symbols that ensures every large interval satisfies the property in Definition~\ref{def:sc}. For this we need to carefully analyze the dependence graph used in the General Lov\'asz Local Lemma. Fortunately it turns out that the same property as in \cite{haeupler2017synchronization} holds: the event that an interval satisfies the property in Definition~\ref{def:sc} is independent of all the corresponding events of disjoint intervals. Similar to \cite{haeupler2017synchronization}, this also gives a randomized algorithm to construct such synchronization strings by using algorithmic versions of the Lov\'asz Local Lemma \cite{moser2010constructive, haeupler2011new}.

To give explicit deterministic constructions of $\varepsilon$-synchronization strings, our starting point is the following observation.\ Suppose we are looking at an interval $S[i, k]$ and $j$ is the midpoint of $i$ and $k$, then we need that $ED(S[i,j],S[j+1,k])>(1-\varepsilon)(k-i)$. This basically means that we need $S[i, j]$ and $S[j+1, k]$ to have large edit distance. Note that $S[i, j]$ and $S[j+1, k]$ have the same length, thus if they are two different codewords of some good codes for insertion and deletion errors, then this property is satisfied. This suggests the following way to construct $\varepsilon$-synchronization strings: take a code for insertion and deletion errors, and concatenate all codewords into a string. We indeed show that with carefully chosen parameters, this idea can work for any interval with relatively large length (e.g., at least twice the length of the codeword). To handle intervals with smaller length, we need to add to every codeword another string with the same length, and this additional string basically corresponds to another synchronization string with the length of the codeword. However, we also need to handle the situation where the interval is split in the middle by the boundary of one codeword and another. For this we introduce a new combinatorial object which we call a \emph{synchronization circle}.\ Intuitively, a synchronization circle is a generalization and strengthening of a synchronization string such that no matter what point one chooses to cut the circle into a string, the resulted string is still a synchronization string. We show how to construct a synchronization circle by concatenating two synchronization strings over different alphabets into a circle. Note that this only doubles the alphabet size. 

Now, we have essentially reduced the task of constructing synchronization strings into finding good codes for insertion and deletion errors. Note that this is the reverse direction of what was established in previous work such as \cite{haeupler2017synchronization}. There, the authors showed how to use synchronization strings to construct good codes for insertion and deletion errors. The connection that we find, thus suggests that synchronization strings and codes for insertion and deletion errors are actually more closely related to each other, in the sense that each of them can be used to construct the other. We view this connection as a main conceptual contribution of our paper.

Going back to the explicit constructions of synchronization strings, suppose we want to construct an $\varepsilon$-synchronization string with length $n$, then all we need is to find a good code for insertion and deletion errors that contain at most $n$ codewords (the number of codewords we need is in fact less than $n$, i.e., $n$ divided by the codeword length). This corresponds to finding a good code for insertion and deletion errors with message length roughly $\log n$. At this point we can use the constructions in \cite{haeupler2017synchronization}, i.e., such a code can be constructed by combining a standard error correcting code with a synchronization string. In fact, we can use our synchronization circle described above as the synchronization string here, since we need to add the synchronization circle to every codeword anyway, and this further saves the alphabet size. We know how to construct a good standard error correcting code efficiently, and now we just need a synchronization string/circle with length $\log n$. This can be done by using a brute-force search which takes only polynomial time for any constant $\varepsilon$. The alphabet we obtain in this way is the concatenation of the alphabet of the small synchronization circle and the alphabet of the error correcting code, and this gives us size $O(\varepsilon^{-3})$. Note that the above argument basically reduces the task of constructing a synchronization string of length $n$ to that of constructing a synchronization string of length $\log n$. Thus to get more efficient construction we can recurse one more time, reducing the task to that of constructing a synchronization string of length $\log \log n$. This way we can achieve near linear time (i.e., $n \poly \log \log n$), but the alphabet size becomes $O(\varepsilon^{-4})$ since we need to use another error correcting code.

\section{Preliminaries}
Usually we use $\Sigma$ (probably with some subscripts) to denote the alphabet.

\begin{definition}[Subsequence] The subsequence of a string $S$ is any sequence of symbols obtained from $S$ by deleting some symbols. It doesn't have to be continuous.
\end{definition}

\begin{definition}[Edit distance] For every $n\in \mathbb{N}$, the edit distance $ED(S,S')$ between two strings $S, S'\in\Sigma^n$ is the minimum number of insertions and deletions required to transform $S$ into $S'$.
\end{definition}

\begin{definition}[Longest Common Subsequence] For any strings $S, S'$ over $\Sigma$, the longest common subsequence of $S$ and $S'$ is the longest pair of subsequence that are equal as strings. We denote by $LCS(S,S')$ the length of the longest common subsequence of $S$ and $S'$.
\end{definition}

Note that $ED(S,S') = |S|+|S'|-2LCS(S,S')$ where $|S|$ denotes the length of $S$.

\begin{definition}[$\varepsilon$-synchronization circle] A string $S$ is
an $\varepsilon$-synchronization circle if  $\forall 1\leq i\leq n$, $S_i,S_{i+1},\dots,S_n,S_1,S_2,\dots,S_{i-1}$ is an $\varepsilon$-synchronization string.
\end{definition}

An $(n, k, d)$ error correcting code (ECC) is a ECC with block length $n$, message length $k$ and distance $d$. For the classic ECC Reed-Solomon code, its encoding can be viewed as a process of multi-point evaluation of a polynomial. The time complexity of the multi-point evaluation is near linear by the following theorem.
\begin{theorem}[Multi-point Evaluation Complexity \cite{borodin1974fast, moenck1972fast, bostan2003tellegen}]
\label{mpetime}
For any $n\in \mathbb{N}$, any  $l = O(\log n) $ s.t. $2^l \geq n$, any polynomial $p$ over  $\mathbb{F}_{2^l}$ of degree at most $n-1$, it takes $O(n \log^2 n)$ arithmetic operations\footnote{Arithmetic operations are $+$ and $ \times$ in the corresponding field.} (including $O(n \log n)$ multiplications)
 to evaluate $p$ on any $n$ points over $\mathbb{F}_{2^l}$.
\end{theorem}
So it immediately follows that the encoding of Reed-Solomon code has time complexity near linear.
\begin{theorem}
\label{RSenctime}
For any $(n, k, d)$ Reed-Solomon code with alphabet size $O(n)$, the encoding takes time $O(n  \log^3 n)$.
\end{theorem}
\begin{proof}
Let's regard the message as $k$ coefficients of a degree $k-1 \leq n-1$ polynomial $p$. To evaluate $p$ on $n$ points it takes $O(n \log^2 n)$ arithmetic operations, including $O(n \log n)$ multiplications by Theorem \ref{mpetime}. Since the field size is $O(n)$, addition operation takes $O(\log n)$ and multiplication takes $O(\log^2 n)$. So the total running time is as stated.

\end{proof}

\section{$\varepsilon$-synchronization Strings and Circles with Alphabet Size $O(\varepsilon^{-2})$}
Now we show that using a non-uniform sample space together with the Lov\'asz Local lemma, we can use a randomized algorithm to construct an $\varepsilon$-synchronization string with alphabet of size $O(\varepsilon^{-2})$, and further we can construct an $\varepsilon$-synchronization circle.

\subsection{Synchronization String}
\begin{lemma}
(General Lov\'asz Local Lemma) Let $A_1,...,A_n$ be a set of bad events. $G(V,E)$ is a dependency graph for this set of events if $V=\{1,\dots,n\}$ and each event $A_i$ is mutually independent of all the events $\{A_j:(i,j)\notin E\}$.

If there exists $x_1,...,x_n\in [0,1)$ such that for all $i$ we have \[\Pr(A_i)\leq x_i\prod_{(i,j)\in E}(1-x_j)\]
Then the probability that none of these events happens is bounded by\[\Pr[\bigwedge_{i=1}^n \bar{A}_i]\geq \prod_{i=1}^n (1-x_i)>0\]
\end{lemma}
\begin{theorem}
\label{syncStr}
$\forall \varepsilon\in(0,1),n\geq 1$, there exists an $\varepsilon$-synchronization string $S$ of length $n$ over alphabet $\Sigma$ of size $\Theta(\varepsilon^{-2})$.
\end{theorem}

\begin{Proof}
Suppose $|\Sigma|=c_1\varepsilon^{-2}$ where $c_1$ is a constant. Let $t=c_2\varepsilon^{-2}$ and $0<c_2<c_1$. The sampling algorithm is as follows:
\begin{enumerate}
  \item Randomly pick $t$ different symbols from $\Sigma$ and let them be the first $t$ symbols of $S$. If $t\geq n$, we just pick $n$ different symbols.
  \item For $t+1\leq i\leq n$, we pick the $i$th symbol $S[i]$ uniformly randomly from $\Sigma\setminus\{S[i-1],\dots,S[i-t+1]\}$
\end{enumerate}
Now we prove that there's a positive probability that $S$ contains no \textit{bad} interval $S[i,k]$ which violates the requirement that $ED(S[i,j],S[j+1,k])>(1-\varepsilon)(k-i)$ for any $i<j<k$. This requirement is equivalent to $LCS(S[i,j],S[j+1,k])< \frac{\varepsilon}{2}(k-i)$.

Notice that for $k-i\leq t$, the symbols in $S[i,k]$ are completely distinct. Hence we only need to consider the case where $k-i>t$. First, let's upper bound the probability that an interval is bad:
\begin{align*}
\Pr [\text{interval I of length } l \text{ is bad}]&\leq\binom{l}{\varepsilon l}(|\Sigma|-t)^{-\frac{\varepsilon l}{2}}\\
&\leq\frac{el}{\varepsilon l}^{\varepsilon l} (|\Sigma|-t)^{-\frac{\varepsilon l}{2}}\\
&\leq(\frac{\varepsilon\sqrt{|\Sigma|-t}}{e})^{-\varepsilon l}\\
& = C^{-\varepsilon l}
\end{align*}
The first inequality holds because if the interval is bad, then it has to contain a repeating sequence $a_1a_2\dots a_pa_1a_2\dots a_p$ where $p$ is at least $\frac{\varepsilon l}{2}$. Such sequence can be specified via choosing $\varepsilon l$ positions in the interval and the probability that a given sequence is valid for the string in this construction is at most $(|\Sigma|-t)^{-\frac{\varepsilon l}{2}}$. The second inequality comes from Stirling's inequality.

The inequality above indicates that the probability that an interval of length $l$ is bad can be upper bounded by $C^{-\varepsilon l}$, where $C$ is a constant and can be arbitrarily large by modifying $c_1$ and $c_2$.

Now we use general Lov\'asz local lemma to show that $S$ contains no bad interval with positive probability. First we'll show the following lemma.
\begin{lemma} The badness of interval $I=S[i,j]$ is mutually independent of the badness of all intervals that do not intersect with $I$.
\end{lemma}

\begin{proofof}{lemma 3.3}
Suppose the intervals before $I$ that do not intersect with $I$ are $I_1,\dots,I_m$, and those after $I$ are $I_1',\dots,I_{m'}'$. We denote the indicator variables of each interval being bad as $b$, $b_k$ and $b_{k'}'$. That is,
\[
b=
\begin{cases}
0 &\text{if $I$ is not bad}\\
1 &\text{if $I$ is bad}
\end{cases}
,\quad b_k=
\begin{cases}
0 &\text{if $I_k$ is not bad}\\
1 &\text{if $I_k$ is bad}
\end{cases}
,\quad b_{k'}'=
\begin{cases}
0 &\text{if $I_{k'}'$ is not bad}\\
1 &\text{if $I_{k'}'$ is bad}
\end{cases}
\]

First we prove that there exists $p\in (0,1)$ such that $\forall x_1,x_2,\dots,x_m\in\{0,1\}$,\[\Pr[b=1|b_k=x_k, k=1,\dots,m]=p\]

According to our construction, we can see that for any fixed prefix $S[1,i-1]$, the probability that $I$ is bad is a fixed real number $p'$. That is, \[\forall \text{ valid }\tilde{S}\in\Sigma^{i-1},\Pr[b=1|S[1,i-1]=\tilde{S}]=p'\]
This comes from the fact that, the sampling of the symbols in $S[i, k]$ only depends on the previous $h=min\{i-1, t-1\}$ different symbols, and up to a relabeling these $h$ symbols are the same $h$ symbols (e.g., we can relabel them as $\{1, \cdots, h\}$ and the rest of the symbols as $\{h+1, \cdots, |\Sigma|\}$). On the other hand the probability that $b=1$ remains unchanged under any relabeling of the symbols, since if two sampled symbols are the same, they will stay the same; while if they are different, they will still be different. Thus we have:
\begin{align*}
&\Pr[b=1|b_k=x_k, i=1,\dots,m]\\
=&\dfrac{\Pr[b=1,b_k=x_k, i=1,\dots,m]}{\Pr[b_k=x_k, k=1,\dots,m]}\\
=&\dfrac{\sum_{\tilde{S}}\Pr[b=1,S[1,i-1]=\tilde{S}]}{\sum_{\tilde{S}}\Pr[S[1,i-1]=\tilde{S}]}\\
=&\sum_{\tilde{S}}(\dfrac{\Pr[b=1,S[1,i-1]=\tilde{S}]}{\Pr[S[1,i-1]=\tilde{S}]}\dfrac{\Pr[S[1,i-1]=\tilde{S}]}{\sum_{\tilde{S}'}\Pr[S[1,i-1]=\tilde{S}']})\\
=&\sum_{\tilde{S}}(\Pr[b=1|S[1,i-1]=\tilde{S}]\dfrac{\Pr[S[1,i-1]=\tilde{S}]}{\sum_{\tilde{S}'}\Pr[S[1,i-1]=\tilde{S}']})\\
=&p'\sum_{\tilde{S}}\dfrac{\Pr[S[1,i-1]=\tilde{S}]}{\sum_{\tilde{S}'}\Pr[S[1,i-1]=\tilde{S}']}\\
=&p'
\end{align*}
In the equations, $\tilde{S}$ indicates all valid string that prefix $S[1,i-1]$ can be such that $b_k=x_k, k=1,\dots,m$. Hence, $b$ is independent of $\{b_k,k=1,\dots,m\}$.
Similarly, we can prove that the joint distribution of $\{b_{k'}',k'=1,\dots,m'\}$ is independent of that of $\{b,b_k,k=1,\dots,m\}$. Hence $b$ is independent of $\{b_k,b_{k'}',k=1,\dots,m,k'=1,\dots,m'\}$, which means, the badness of interval $I$ is mutually independent of the badness of all intervals that do not intersect with $I$. $\square$
\end{proofof}

Obviously, an interval of length $l$ intersects at most $l+l'$ intervals of length $l'$. To use Lov\'asz local lemma, we need to find a sequence of real numbers $x_{i,k}\in[0.1)$ for intervals $S[i,k]$ for which
\[\Pr[S[i,k]\text{is bad}]\leq x_{i,k}\prod_{S[i,k]\cap S[i',k']\neq\emptyset}(1-x_{i',k'})\]

The rest of the proof is the same as that of Theorem 5.7 in \cite{haeupler2017synchronization}.

We propose $x_{i,k}=D^{-\varepsilon (k-i)}$ for some constant $D\geq 1$. Hence we only need to find a constant $D$ such that for all $S[i,k]$,
\[C^{-\varepsilon(k-i)}\leq D^{-\varepsilon (k-i)}\prod_{l=t}^n[1-D^{-\varepsilon l}]^{l+(k-i)}\]
That is, for all $l'\in\{1,...,n\}$,
\[C^{-l'}\leq D^{-l'}\prod_{l=t}^n[1-D^{-\varepsilon l}]^{\frac{l+l'}{\varepsilon}}\]
which means that
\[C\geq \dfrac{D}{\prod_{l=t}^n[1-D^{-\varepsilon l}]^{\frac{l/l'+1}{\varepsilon}}}\]
Notice that the righthand side is maximized when $n=\infty, l'=1$. Hence it's sufficient to show that
\[C\geq \dfrac{D}{\prod_{l=t}^{\infty}[1-D^{-\varepsilon l}]^{\frac{l+1}{\varepsilon}}}\]
Let $L=\max_{D>1}\dfrac{D}{\prod_{l=t}^{\infty}[1-D^{-\varepsilon l}]^{\frac{l+1}{\varepsilon}}}$. We only need to guarantee that $C>L$.

We claim that $L=\Theta(1)$. Since that $t=c_2\varepsilon^{-2}=\omega(\frac{\log\frac{1}{\varepsilon}}{\varepsilon})$,
\begin{align}
\dfrac{D}{\prod_{l=t}^{\infty}[1-D^{-\varepsilon l}]^{\frac{l+1}{\varepsilon}}}&<\dfrac{D}{\prod_{l=t}^{\infty}[1-\frac{l+1}{\varepsilon}D^{-\varepsilon l}]}\\
&<\dfrac{D}{1-\sum_{l=t}^{\infty}\frac{l+1}{\varepsilon}D^{-\varepsilon l}}\\
&=\dfrac{D}{1-\frac{1}{\varepsilon}\sum_{l=t}^{\infty}(l+1)D^{-\varepsilon l}}\\
&=\dfrac{D}{1-\frac{1}{\varepsilon}\frac{2tD^{-\varepsilon t}}{(1-D^{-\varepsilon})^2}}\\
&=\dfrac{D}{1-\frac{2}{\varepsilon^3}\frac{D^{-\frac{1}{\varepsilon}}}{(1-D^{-\varepsilon})^2}}
\end{align}

Inequality $(1)$ comes from the fact that $(1-x)^{\alpha}>1-\alpha x$, $(2)$ comes from he fact that $\prod_{i=1}^{\infty}(1-x_i)\geq 1-\sum_{i=1}^{\infty}x_i$ and $(3)$ is a result from $\sum_{l=t}^{\infty}(l+1)x^l=\frac{x^t(1+t-tx)}{(1-x)^2}<\frac{2tx^t}{(1-x)^2},x<1$.

We can see that for $D=7$, $\max_{\varepsilon}\{\frac{2}{\varepsilon^3}\frac{D^{-\frac{1}{\varepsilon}}}{(1-D^{-\varepsilon})^2}\}<0.9$. Therefore (5) is bounded by a constant, which means $L=\Theta(1)$ and the proof is complete.
\end{Proof}

\begin{lemma}
There exists a randomized algorithm s.t. for any $\varepsilon\in(0,1)$ and any $ n\in \mathbb{N} $, it can construct an $\varepsilon$-synchronization string of length $n$ over alphabet of size $O(\varepsilon^{-2})$ in expected time $O(n^5\log n)$.
\end{lemma}

\begin{proof}

The algorithm is similar to that of Lemma 5.8 in \cite{haeupler2017synchronization}, using algorithmic Lov\'asz Local lemma \cite{moser2010constructive} and the extension in \cite{haeupler2011new}. It starts with a string sampled according to the sampling algorithm in the proof of Theorem \ref{syncStr}, over alphabet $\Sigma$ of size $C\varepsilon^{-2}$ for some large enough constant $C$.\ Then the algorithm checks all $O(n^2)$ intervals for a violation of the requirements for $\varepsilon$-synchronization string.\ If a bad interval is found, this interval is re-sampled by randomly choosing every symbol s.t. each one of them is different from the previous $t-1$ symbols, where $t = c'\varepsilon^{-2}$ with $c'$ being a constant smaller than $C$.

One subtle point of our algorithm is the following. Note that in order to apply the algorithmic framework of \cite{moser2010constructive} and \cite{haeupler2011new}, one needs the probability space to be sampled from $n$ independent random variables ${\cal P}=\{P_1, \cdots, P_n\}$ so that each event in the collection ${\cal A}=\{A_1, \cdots, A_m\}$ is determined by some subset of $\cal P$. Then, when some bad event $A_i$ happens, one only resamples the random variables that decide $A_i$. Upon first look, it may appear that in our application of the Lov\'asz Local lemma, the sampling of the $i$'th symbol depends on the the previous $h=min\{i-1, t-1\}$ symbols, which again depend on previous symbols, and so on. Thus the sampling of the $i$'th symbol depends on the sampling of all previous symbols.\ However, we can implement our sampling process as follows: for the $i$'th symbol we first independently generate a random variable $P_i$ which is uniform over $\{1, 2, \cdots, |\Sigma|-h\}$, then we use the random variables $\{P_1, \cdots, P_n\}$ to decide the symbols, in the following way. Initially we fix some arbitrary order of the symbols in $\Sigma$, then for $i=1, \cdots, n$, to get the $i$'th symbol, we first reorder the symbols $\Sigma$ so that the previous $h$ chosen symbols are labeled as the first $h$ symbols in $\Sigma$, and the rest of the symbols are ordered in the current order as the last $|\Sigma|-h$ symbols. We then choose the $i$'th symbol as the $(h+P_i)$'th symbol in this new order. In this way, the random variables $\{P_1, \cdots, P_n\}$ are indeed independent, and the $i$'th symbol is indeed chosen uniformly from the $|\Sigma|-h$ symbols excluding the previous $h$ symbols. Furthermore, the event of any interval $S[i, k]$ being bad only depends on the random variables $(P_i, \cdots, P_k)$ since no matter what the previous $h$ symbols are, they are relabeled as $\{1, \cdots, h\}$ and the rest of the symbols are labeled as $\{h+1, \cdots, |\Sigma|\}$. From here, the same sequence of $(P_i, \cdots, P_k)$ will result in the same behavior of $S[i, k]$ in terms of which symbols are the same. We can thus apply the same algorithm as in \cite{haeupler2017synchronization}.

Note that the time to get the $i$'th symbol from the random variables $\{P_1, \cdots, P_n\}$ is $O(n \log \frac{1}{\varepsilon})$ since we need $O(n)$ operations each on a symbol of size $C\varepsilon^{-2}$. Thus resampling each interval takes $O(n^2 \log \frac{1}{\varepsilon})$ time since we need to resample at most $n$ symbols.  For every interval, the edit distance can be computed using the Wagner-Fischer dynamic programming within $O(n^2 \log \frac{1}{\varepsilon})$ time. \cite{haeupler2011new} shows that the expected number of re-sampling is $O(n)$. The algorithm will repeat until no bad interval can be found.  Hence the overall expected running time is $O(n^5 \log \frac{1}{\varepsilon})$.

Note that without loss of generality we can assume that $\varepsilon>
1/\sqrt{n}$ because for smaller errors we can always use the indices
directly, which have alphabet size $n$. So the overall expected running time is $O(n^5 \log n)$.
\end{proof}
\subsection{Synchronization circle}
We now construct an $\varepsilon$-synchronization circle  using Theorem \ref{syncStr}.
\begin{theorem}
\label{syncCircle}
For every $  \varepsilon\in(0,1), n \in \mathbb{N}$, there exists an $\varepsilon$-synchronization circle $S$ of length $n$ over alphabet $\Sigma$ of size $\Theta(\varepsilon^{-2})$.
\end{theorem}

\begin{proof}
First, by Theorem \ref{syncStr}, we can have two $\varepsilon$-synchronization strings: $S_1$ with length $\lceil \frac{n}{2}\rceil$ over $\Sigma_1$ and $S_2$ with length $\lfloor\frac{n}{2}\rfloor$ over $\Sigma_2$. Let $\Sigma_1\cap\Sigma_2=\emptyset$ and $|\Sigma_1|=|\Sigma_2|= O(\varepsilon^{-2})$. Let $S$ be the concatenation of $S_1$ and $S_2$. Then $S$ is over alphabet $\Sigma=\Sigma_1\cup\Sigma_2$ whose size is $O(\varepsilon^{-2})$. Now we prove that $S$ is an $\varepsilon$-synchronization circle.

$\forall 1\leq m\leq n$, consider  string $S'=s_m,s_{m+1},\dots,s_n,s_1,s_2,\dots,s_{m-1}$. Notice that for two strings $T$ and $T'$ over alphabet $\Sigma$, $LCS(T,T')\leq\frac{\varepsilon}{2}(|T|+|T'|)$ is equivalent to $ED(T,T')\geq(1-\varepsilon)(|T|+|T'|)$. For any $i<j<k$, we call an interval $S'[i,k]$ good if $LCS(S'[i,j], S'[j+1,k])\leq \frac{\varepsilon}{2}(k-i)$. It suffices to show that $\forall 1\leq i,k\leq n$, the interval $S'[i,k]$ is good.


Without loss of generality let's assume $ m \in [\lceil\frac{n}{2}\rceil, n ]$.

Intervals which are substrings of $S_1$ or $S_2$ are good intervals, since $S_1$ and $ S_2$ are $\varepsilon$-synchronization strings.

We are left with  intervals crossing the ends of $S_1$ or $S_2$.
\paragraph{If $S'[i,k]$ contains $s_n,s_1$ but doesn't contain $s_{\lceil\frac{n}{2}\rceil}$:} If $j< n-m+1$, then there's no common subsequence between $s'[i,j]$ and $S'[n-m+2,k]$. Thus \[LCS(S'[i,j],S'[j+1,k])\leq LCS(S'[i,j],S'[j+1,n-m+1])\leq \frac{\varepsilon}{2}(n-m+1-i)<\frac{\varepsilon}{2}(k-i)\]
If $j\geq n-m+1$, then there's no common subsequence between $S'[j+1,k]$ and $S'[i,n-m+1]$. Thus
\[LCS(S'[i,j],S'[j+1,k])\leq LCS(S'[n-m+2,j],S'[j+1,k])\leq\frac{\varepsilon}{2}(k-(n-m+2))<\frac{\varepsilon}{2}(k-i)\]
Thus intervals of this kind are good.

\begin{figure}[H]
  \centering
  \includegraphics[width=8cm]{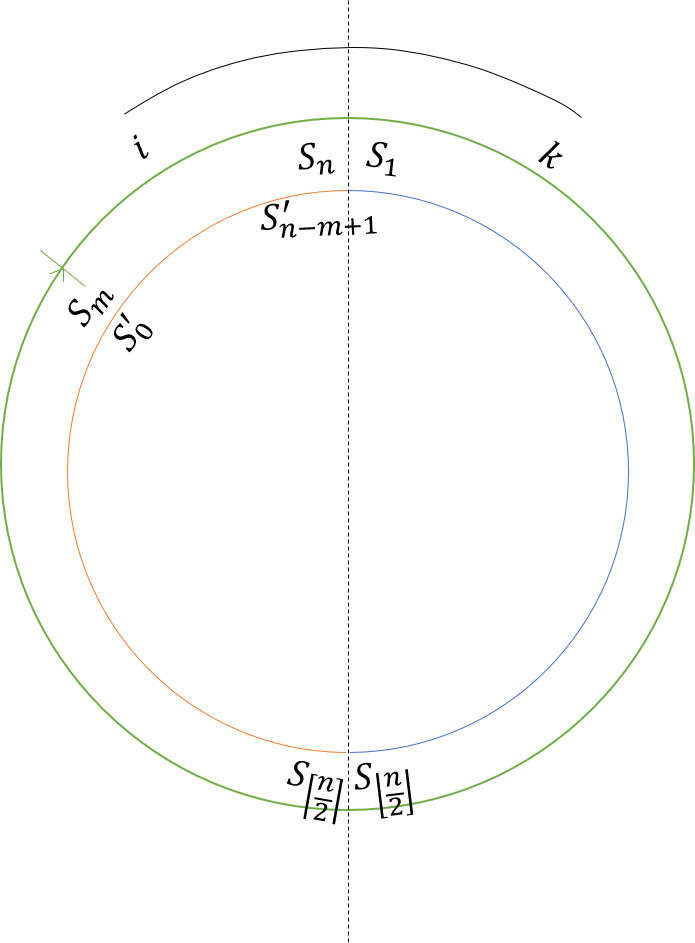}\\
  \caption{Example where $S'[i,k]$ contains $s_n,s_1$ but doesn't contain $s_{\lceil\frac{n}{2}\rceil}$}\label{sc1}
\end{figure}
\paragraph{If $S'[i,k]$ contains $s_{\lfloor\frac{n}{2}\rfloor},s_{\lceil\frac{n}{2}\rceil}$ but doesn't contain $s_n$:} If $j\leq n-m+\lfloor\frac{n}{2}\rfloor+1$, then there's no common subsequence between $S'[i,j]$ and $S'[n-m+\lceil\frac{n}{2}\rceil+1,k]$, thus
\[LCS(S'[i,j],S'[j+1,k])\leq LCS(S'[i,j],S'[j+1,n-m+\lfloor\frac{n}{2}\rfloor+1])<\frac{\varepsilon}{2}(k-i)\]
If $j\geq n-m+\lfloor\frac{n}{2}\rfloor+1$, then there's no common subsequence between $S'[j+1,k]$ and $S'[i,n-m+\lfloor\frac{n}{2}\rfloor+1]$. Thus
\[LCS(S'[i,j],S'[j+1,k])\leq LCS(S'[n-m+\lceil\frac{n}{2}\rceil+1,j],S'[j+1,k])<\frac{\varepsilon}{2}(k-i)\]
Thus intervals of this kind are good.

\begin{figure}[H]
  \centering
  \includegraphics[width=8cm]{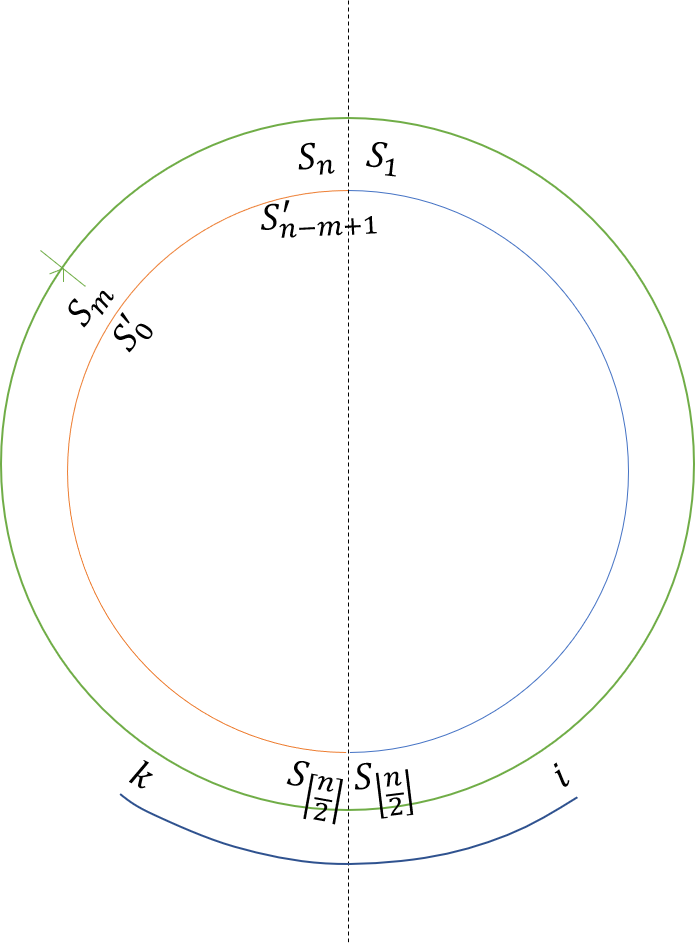}\\
  \caption{Example where $S'[i,k]$ contains $s_{\lfloor\frac{n}{2}\rfloor},s_{\lceil\frac{n}{2}\rceil}$}\label{sc2}
\end{figure}
\paragraph{If $S'[i,k]$ contains $s_{\lceil\frac{n}{2}\rceil}$ and $s_n$:} If $n-m+2\leq j\leq n-m+\lfloor\frac{n}{2}\rfloor+1$, then the common subsequence is either that of $S'[i,n-m+1]$ and $S'[n-m+\lceil\frac{n}{2}\rceil+1,k]$ or that of $S'[n-m+2,j]$ and $S'[j+1,n-m+\lfloor\frac{n}{2}\rfloor+1]$.
This is because $\Sigma_1\cap\Sigma_2=\emptyset$.
Thus
\begin{align*}
&LCS(S'[i,j],S'[j+1,k])\\
\leq&\max\{LCS(S'[i,n-m+1],S'[n-m+\lceil\frac{n}{2}\rceil+1,k]),\\
&\qquad LCS(S'[n-m+2,j],S'[j+1,n-m+\lfloor\frac{n}{2}\rfloor+1])\}\\
<&\frac{\varepsilon}{2}(k-i)
\end{align*}
\begin{figure}[H]
  \centering
  \includegraphics[width=8cm]{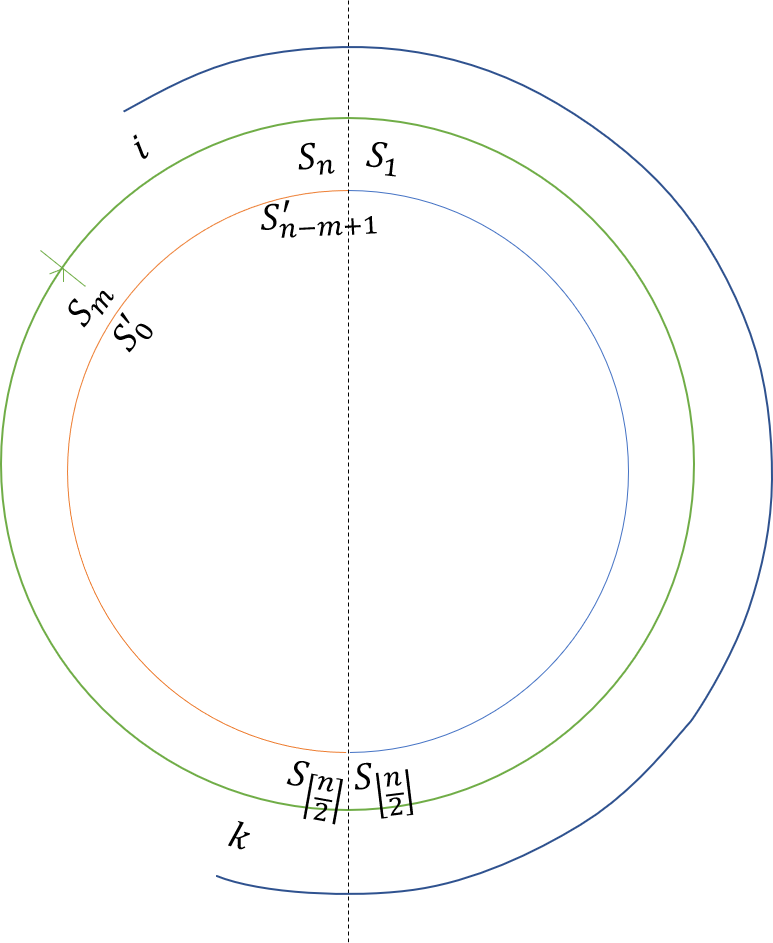}\\
  \caption{Example where $S'[i,k]$ contains $s_{\lceil\frac{n}{2}\rceil}$ and $s_n$}\label{sc3}
\end{figure}
If $j \leq n-m+1$, then there's no common subsequence between $S'[i,j]$ and $S'[n-m+2,n-m+\lfloor\frac{n}{2}\rfloor+1]$. Thus
\begin{align*}
&LCS(S'[i,j],S'[j+1,k])\\
\leq &LCS(S'[i,j],S'[j+1,n-m+1])+LCS( S'[i,j], S'[n-m+\lceil\frac{n}{2}\rceil+1,k])\\
< &    \frac{\varepsilon}{2}( n-m+1 - i  ) + \frac{\varepsilon}{2}( n - \lceil \frac{n}{2} \rceil )      \\
\leq &   \frac{\varepsilon}{2}(  n-m+1 - i  ) + \frac{\varepsilon}{2}(k- (n-m+2))      \\
= &   \frac{\varepsilon}{2}(k-1-i) \\
< &   \frac{\varepsilon}{2}(k-i)
\end{align*}
If $j\geq S'[n-m+\lceil\frac{n}{2}\rceil+1]$, the proof is similar to the case where $j\leq n-m+1$.


This shows that $S'$ is an $\epsilon$-synchronization string. Thus by the definition of synchronization circle, the construction gives an  $\epsilon$-synchronization circle.

\end{proof}

\section{Deterministic Constructions}
We now construct $\varepsilon$-synchronization strings using synchronization circles. First we recall the following result from  \cite{haeupler2017synchronization}.

\begin{lemma}[Theorem 4.2 of \cite{haeupler2017synchronization}]
\label{insdelCode}
Given an $\varepsilon$-synchronization string $S$ with length $n$, alphabet $\Sigma_S$, and an efficient ECC $\mathcal{C}$ with block length $n$, alphabet $\Sigma_C$, that corrects up to $n\delta \frac{1+\varepsilon}{1-\varepsilon}$ half-errors, one can obtain an insertion/deletion code $\mathcal{C}'$  that can be decoded from up to $n\delta$ deletions, where $ \mathcal{C}' = \{(c'_1, \ldots, c'_n)| \forall  i\in [n],  c'_i = (c_i, S[i]), (c_1,\ldots, c_n) \in \mathcal{C}  \} $.

\end{lemma}
Using the insertion deletion code in this lemma, we have the following construction of a synchronization circle $S$ of length $n$.

\begin{algo}
\label{algo}
Main Construction.

For every $n, m \in \mathbb{N}, m\leq n$, we have the following.

\textbf{Input:}
\begin{enumerate}
  \item A ECC $\tilde{\mathcal{C}}\subset\Sigma_{\tilde{c}}^m$, with distance $ \delta m $ and $|\tilde{\mathcal{C}}| \geq \ell = \lceil \frac{n}{m}\rceil$.
  \item An $\varepsilon$-synchronization circle $SC$ of length $m$ over alphabet $\Sigma_{sc}$ i.e. $ SC=(sc_1,sc_2,\dots,sc_m) \in\Sigma_{sc}^m$.
\end{enumerate}

\textbf{Output:} An $ \varepsilon'$-synchronization  $S$ circle of length $n$.

\textbf{Operations:}

\begin{itemize}
  \item Construct a new code $\mathcal{C}\subset\Sigma^m$ s.t.
$$ \mathcal{C} = \{c=((\tilde{c}_1,sc_1),(\tilde{c}_2,sc_2),\dots,(\tilde{c}_m,sc_m))|(\tilde{c}_1,\tilde{c}_2,\dots,\tilde{c}_m)\in\tilde{\mathcal{C}}\},$$
where $\Sigma=\Sigma_{\tilde{c}}\times\Sigma_{sc}$.
  \item Choose $\ell $ codewords $C_1,C_2,\dots,C_{\ell}$ from $\mathcal{C}$.
  \item Let $S$ be concatenation of these codewords: $S=C_1\circ C_2\circ\dots\circ C_{\ell}$.
\end{itemize}

\end{algo}

\begin{lemma}
\label{main}
The output $S$ in Algorithm \ref{algo} is an $\varepsilon'$-synchronization circle, where $\varepsilon' \leq 10(1-\frac{1-\varepsilon}{1+\varepsilon}\delta) $.
\end{lemma}

\begin{proof}
Suppose $\tilde{\mathcal{C}}$ can correct up to $\delta m$ half-errors. Then according to lemma \ref{insdelCode}, $\mathcal{C}$ can correct up to $\frac{1-\varepsilon}{1+\varepsilon} \delta m$ deletions.

Let $\alpha=1-\frac{1-\varepsilon}{1+\varepsilon}\delta$. Notice that $\mathcal{C}$ has the following properties:
\begin{enumerate}
  \item $LCS(\mathcal{C}) = \max_{c_1, c_2\in C}{LCS(c_1, c_2)} \leq \alpha m$
  \item Each codeword in $\mathcal{C}$ is an $\varepsilon$-synchronization circle over $\Sigma$.
\end{enumerate}

Consider any shift of the start point of $S$, we only need to prove that $\forall 1\leq i < j < k\leq n, LCS(S[i,j],S[j+1,k])<\frac{\varepsilon'}{2}(k-i)$. First we prove the lemma below.

\begin{lemma}
\label{lcs}
Suppose $T_1$ is the concatenation of $\ell_1$ strings, $T_1=S_1\circ\dots\circ S_{\ell_1}$ and $T_2$ is the concatenation of $\ell_2$ strings, $T_2=S'_1\circ\dots\circ S'_{\ell_2}$. If  there exists an integer $t$ such that for all $i, j$, we have $LCS(S_i, S'_j) \leq t$, then we have $LCS(T_1, T_2) \leq (\ell_1+\ell_2)t$.
\end{lemma}

\begin{proof}[Proof of \ref{lcs}]
We rename the strings in $T_1$ by $S_1, \cdots, S_{\ell_1}$ and rename the strings in $T_2$ by $S_{\ell_1+1}, \cdots, S_{\ell_1+\ell_2}$. Suppose the longest common subsequence between $T_1$ and $T_2$ is $\tilde{T}$, which can be viewed as a matching between $T_1$ and $T_2$.

we can divide $\tilde{T}$ sequentially into disjoint intervals, where each
interval corresponds to a common subsequence between a different pair
of strings $(S_i, S_j)$, where $S_i$ is from $T_1$ and $S_j$ is from $T_2$. In
addition, if we look at the intervals from left to right, then for any
two consecutive intervals and their corresponding pairs $(S_i, S_j)$ and $(S_{i'}, S_{j'})$, we must have $i' \geq i$ and $j' \geq j$  since the matchings which correspond to two intervals cannot cross each other. Furthermore either $i' > i$
or $j' > j$ as the pair $(S_i, S_j)$ is different from $(S_{i'}, S_{j'})$. 

Thus, starting from the first interval, we can
label each interval with either $i$ or $j $ such that every interval
receives a different label, as follows. We label the first interval using either $i$ or $j$. Then, assuming we have already labeled some intervals and now look at the next interval. Without loss of generality assume that the previous interval is labeled using $i$, now if the current $i' > i$ then we can label the current interval using $i'$; otherwise we must have $j'>j$ so we can label the current interval using $j'$. Thus the total number of the labels is at
most $l_1+l_2$, which means the total number of the intervals is also at
most $l_1+l_2$. Note that each interval has length at most $t$, therefore we can upper bound $LCS(T_1, T_2)$ by
 $(l_1+l_2)t$.

%
%

\end{proof}

Suppose $S_1=S[i,j]$ and $S_2=S[j+1,k]$. Let $\varepsilon' = 10\alpha$.
\paragraph{Case 1:} $k-i>m$. Let $|S_1|=s_1$ and $|S_2|=s_2$, thus $s_1+s_2 > m$. If we look at each $S_h$ for $h=1, 2$, then $S_h$ can be divided into some consecutive codewords, plus at most two incomplete codewords at both ends. In this sense each $S_h$ is the concatenation of $\ell_h$ strings with $\ell_h < \frac{s_h}{m}+2$. An example of the worst case appears in Figure \ref{4}.

\begin{figure}[H]
  \centering
  \includegraphics[width=8cm]{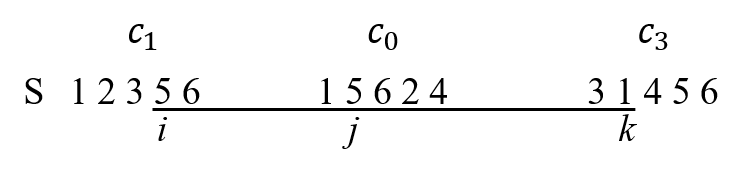}\\
  \caption{Example of the worst case, where $j$ splits a codeword, and there are two incomplete codewords at both ends.}\label{4}
\end{figure}

Now consider the longest common subsequence between any pair of these strings where one is from $S_1$ and the other is from $S_2$, we claim that the length of any such longest common subsequence is at most $\alpha m$. Indeed, if the pair of strings are from two different codewords, then by the property of the code $\mathcal{C}$ we know the length is at most $\alpha m$. On the other hand, if the pair of strings are from a single codeword (this happens when $j$ splits a codeword, or when $S[i]$ and $S[k]$ are in the same codeword), then they must be two disjoint intervals within a codeword. In this case, by the property that any codeword is also a synchronization circle, the length of the longest common subsequence of this pair is at most $\frac{\varepsilon}{2} m$.

Note that $\alpha=1-\frac{1-\varepsilon}{1+\varepsilon}\delta \geq  1-\frac{1-\varepsilon}{1+\varepsilon} =\frac{2\varepsilon}{1+\varepsilon} \geq \varepsilon$ (since $\delta, \varepsilon \in (0, 1)$). Thus $\frac{\varepsilon}{2} m < \alpha m$. Therefore, by Lemma~\ref{lcs}, we have 

\begin{align*}
& LCS(S_1, S_2) \\ 
< &(\frac{s_1}{m}+2+\frac{s_2}{m}+2) \alpha m \\
= &\alpha (s_1+s_2+4 m) \\
< & 5 \alpha (s_1+s_2) \\
= & 5 \alpha (k-i) = \frac{\varepsilon'}{2}(k-i)
\end{align*}

\paragraph{Case 2:} If $k-i\leq m$, then according to the property of synchronization circle $ SC $, we know that the longest common subsequence of $S_1$ and $S_2$ is less than $\frac{\varepsilon}{2}(k-i)\leq \alpha(k-i)\leq\frac{\varepsilon'}{2}(k-i)$.

As a result, the longest common subsequence of $S[i,j]$ and $S[j+1,k]$ is less than $\frac{\varepsilon'}{2}(k-i)$, which means that $S$ is an $\varepsilon'$-synchronization circle. 
\end{proof}

\begin{lemma}
\label{smallcode}

For any $n\in \mathbb{N}$, any $\varepsilon \in [0,1]$,
one can construct a ECC in time $  O(2^{\varepsilon n} (\frac{2e}{\varepsilon})^n n \log(1/\varepsilon))  $ and space $  O(2^{\varepsilon n} n \log (1/\varepsilon) ) $, with block length $n$, number of codewords $  2^{\varepsilon n}$, distance $ d = (1-\varepsilon)n$, alphabet size $2e/\varepsilon$.
\end{lemma}

\begin{proof}

We conduct a brute-force search here to find all the codewords one by one.

We denote the code as $\mathcal{C}$ and the alphabet as $\Sigma$. Let $|\Sigma| = q$.
At first, let $\mathcal{C} = \emptyset$. Then we add an arbitrary element in $\Sigma^n$ to $\mathcal{C}$. Every time after a new element $C$ is added to $\mathcal{C}$, we exclude every such element in $\Sigma^{n}$ that has distance less than $d$ from $C$. Then we pick an arbitrary one from the remaining elements, adding it to $\mathcal{C}$. Keep doing this until $|\mathcal{C}| = 2^{\varepsilon n}$.

Note that given $C \in \Sigma^n$, the total number of elements that have distance less than $d$ to $C$, is at most ${n \choose d} q^d  = {n \choose (n-d)} q^d \leq (\frac{e}{\varepsilon})^{\varepsilon n} q^{(1-\varepsilon) n} $. We have to require that  $ |\mathcal{C}| (\frac{e}{\varepsilon})^{\varepsilon n} q^{(1-\varepsilon) n} \leq q^n $.
Let $q = 2e/\varepsilon$. So $\mathcal{C}$ can be $2^{\varepsilon n}$.

The exclusion operation takes time $  O((\frac{2e}{\varepsilon})^n n \log(1/\varepsilon)) $ as we have to
exhaustively search the space and for each word we have to compute it's hamming distance to the new added code word. Since there are  $  2^{\varepsilon n}$ code words, the time complexity is as stated.

We have to record those code words, so the space complexity is also as stated.
\end{proof}

\begin{theorem}
\label{mainTheorem}
For  any $n\in \mathbb{N}$ and any $\varepsilon \in (0, 1)$, an $\varepsilon$-synchronization circle $S$ of length $n$ over alphabet $\Sigma$ where $|\Sigma| = O(\varepsilon^{-3})$ can be constructed in time  $(O(\frac{1}{\varepsilon^2}))^{O(\frac{\log n}{\varepsilon})}\cdot \poly (\frac{\log n}{\varepsilon})$.  If $\varepsilon$ is a constant, then the running time is $\poly(n)$.
\end{theorem}
\begin{proof}
By Lemma \ref{smallcode}, we can construct a  ECC $ \mathcal{C} $ with block length $m = O(\frac{\log n}{\varepsilon})$, $|\mathcal{C}  | = n$, distance $\rho m$,  $\rho=\frac{1+\frac{\varepsilon}{30}}{1-\frac{\varepsilon}{30}}(1-\frac{\varepsilon}{10})=1-\Omega(\varepsilon)$,  alphabet $\Sigma_{ \mathcal{C} }$ of size $O(1/\varepsilon)$. Let $\mathcal{SC}$ be an $\frac{\varepsilon}{30}$-synchronization circle over alphabet $\Sigma_{sc}$ with length $m$. Let $|\Sigma_{sc}| = O(\varepsilon^{-2})$. Then according to the construction algorithm and lemma \ref{main}, we have an $\varepsilon$-synchronization string $S$.

The construction of  $ SC $ takes $(O(\frac{1}{\varepsilon^2}))^{O(m)}\cdot \poly (m)=(O(\frac{1}{\varepsilon^2}))^{O(\frac{\log n}{\varepsilon})}\cdot \poly (\frac{\log n}{\varepsilon})$. By Lemma \ref{smallcode}, constructing $\mathcal{C}$ takes time $ O(2^{\varepsilon m} (\frac{2e}{\varepsilon})^m m \log(1/\varepsilon))  $. Thus the total running time is $(O(\frac{1}{\varepsilon^2}))^{O(\frac{\log n}{\varepsilon})}\cdot \poly (\frac{\log n}{\varepsilon})$. Regarding $\varepsilon$ as a constant, it is $\poly(n)$.
\end{proof}

To make the construction more efficient, we can use a 2-level recursion to reduce the running time to $n\cdot \poly\log n$, which is near linear.

\begin{theorem}
\label{efficient}
There exists a constant $C>1$ such that for any $n\in \mathbb{N}$ and any $\varepsilon \geq \frac{C (\log \log n)^2}{\log n}$, an $\varepsilon$-synchronization circle $S$ over alphabet $\Sigma$ where $|\Sigma| = O(\varepsilon^{-4})$ can be constructed in $O(n\cdot  (\log \log n)^2)$ time.
\end{theorem}

\begin{proof}
First by Theorem \ref{mainTheorem} we can construct a synchronization circle $S_0$ of length $m' = O(\frac{\log n (\log \log n + \log \frac{1}{\varepsilon}) }{\varepsilon^2} )$ over alphabet $\Sigma_0$ of size $O(\varepsilon^{-3})$. Let $ECC$ $\mathcal{C}$ be the concatenation code of two codes, an outer $(m=O(\frac{\log n}{\varepsilon}),  \Omega(\varepsilon m) , (1-O(\varepsilon))m)$ Reed-Solomon code with alphabet size $m$, and an inner code by Lemma \ref{smallcode} with block length $m_0 = O(\frac{\log m}{\varepsilon})$, number of codewords $m$, distance $(1-O(\varepsilon))m_0$,  alphabet size $O(1/\varepsilon)$. Thus $\mathcal{C}$ is an ECC with block length $m' =m m_0$, number of code words at least $n$, distance  $(1-O(\varepsilon))m'$, alphabet size $O(1/\varepsilon)$. Then one can use $S_0$ together with $\mathcal{C}$ to construct an $\varepsilon$-synchronization string over alphabet $\Sigma$ of size $O(\varepsilon^{-4})$, according to Algorithm \ref{algo}.

By Theorem \ref{mainTheorem}, the construction of $S_0$ takes $(O(  \frac{1}{\varepsilon^2} ))^{O(\frac{\log m'}{\varepsilon})}\cdot \poly (\frac{\log m'}{\varepsilon})$. The time to copy it for $n$ positions is $O(n \log(1/\varepsilon))$. The time to compute the inner code is $m(\frac{2e}{\varepsilon})^{O(m_0)} O(m_0 \log \frac{1}{\varepsilon})$. The inner code only have to be computed once. Computing the outer Reed-Solomon code takes time $O(m \log^3 m)$ by Theorem \ref{RSenctime}.
Copying the inner code for all symbols of one
outer codeword takes time $O(m m_0 \log(1/\varepsilon))$, so the total
time to compute one codeword of the concatenated code is $O(m \log^3m+ m m_0 \log(1/\varepsilon))$. We use $n/m'=n/(m m_0)$ concatenated codewords for  $S$, so this takes time
$\frac{n}{m m_0}(O(m \log^3m+ m m_0 \log(1/\varepsilon)))$.
Thus the total running time is
\begin{equation}
\begin{split}
&(O(  \frac{1}{\varepsilon^2}))^{O(\frac{\log m'}{\varepsilon})}\cdot \poly (\frac{\log m'}{\varepsilon}) + O(n \log(1/\varepsilon)) \\
+ & m(\frac{2e}{\varepsilon})^{O(m_0)} O(m_0 \log \frac{1}{\varepsilon})
 +  \frac{n}{m_0}( O( \log^3 m+ m_0 \log(1/\varepsilon)))\\
= &  (O(\frac{1}{\varepsilon}))^{O(\frac{\log \log n+\log(1/\varepsilon)}{\varepsilon})}+O(n \log(1/\varepsilon))+ O(\varepsilon n  \log^2 \frac{\log n}{\varepsilon}).
\end{split}
\end{equation}

Thus as long as $\varepsilon \geq \frac{C (\log \log n)^2}{\log n}$ for some constant $C>1$,  the running time is then $O(n\cdot  (\log \log n)^2)$, which is near linear.

%
%
\end{proof}

\section{Discussion and Open Problems}
One question here is to improve the alphabet size in our deterministic constructions. For example, it would be good to give an efficient deterministic construction that matches Theorem~\ref{syncStr}. It seems a little tricky to achieve this using our approach, but other more sophisticated approaches may work (e.g., direct derandomization of Theorem~\ref{syncStr}).

Perhaps more interestingly, our work shows another connection between synchronization strings and codes for insertion and deletion errors --- the latter can be used to construct the former. This can be viewed as the reverse direction of the connection found in \cite{haeupler2017synchronization}. Together these results show that synchronization strings and codes for insertion and deletion errors are closely related objects. Thus any improvement to one may lead to improvement to the other. Indeed, these two objects are similar in several aspects. For example, when considering $\varepsilon$-synchronization strings and codes that can correct $1-\varepsilon$ fraction of insertion and deletion errors, both of them have alphabet size $O(1/\varepsilon^2)$ when using a random construction. However, we note that \cite{BukhV16} constructed codes that can correct $1-\varepsilon$ fraction of insertion and deletion errors with alphabet size $O(1/\varepsilon)$, which beats the random construction. Thus it is a natural and interesting open problem to see if there also exist $\varepsilon$-synchronization strings with alphabet size $O(1/\varepsilon)$.

\bibliographystyle{plain}
\bibliography{reference}

\begin{thebibliography}{10}

\bibitem{borodin1974fast}
Allan Borodin and Robert Moenck.
\newblock Fast modular transforms.
\newblock {\em Journal of Computer and System Sciences}, 8(3):366--386, 1974.

\bibitem{bostan2003tellegen}
Alin Bostan, Gr{\'e}goire Lecerf, and {\'E}ric Schost.
\newblock Tellegen's principle into practice.
\newblock In {\em Proceedings of the 2003 international symposium on Symbolic
  and algebraic computation}, pages 37--44. ACM, 2003.

\bibitem{braverman2017coding}
Mark Braverman, Ran Gelles, Jieming Mao, and Rafail Ostrovsky.
\newblock Coding for interactive communication correcting insertions and
  deletions.
\newblock {\em IEEE Transactions on Information Theory}, 63(10):6256--6270,
  2017.

\bibitem{BukhV16}
Boris Bukh and Venkatesan Guruswami.
\newblock An improved bound on the fraction of correctable deletions.
\newblock In {\em Proceedings of the twenty-seventh annual ACM-SIAM symposium
  on Discrete algorithms}, pages 1893--1901. ACM, 2016.

\bibitem{gelles2015coding}
Ran Gelles.
\newblock Coding for interactive communication: A survey, 2015.

\bibitem{gelles2015capacity}
Ran Gelles and Bernhard Haeupler.
\newblock Capacity of interactive communication over erasure channels and
  channels with feedback.
\newblock In {\em Proceedings of the Twenty-Sixth Annual ACM-SIAM Symposium on
  Discrete Algorithms}, pages 1296--1311. Society for Industrial and Applied
  Mathematics, 2015.

\bibitem{ghaffari2014optimal2}
Mohsen Ghaffari and Bernhard Haeupler.
\newblock Optimal error rates for interactive coding ii: Efficiency and list
  decoding.
\newblock In {\em Foundations of Computer Science (FOCS), 2014 IEEE 55th Annual
  Symposium on}, pages 394--403. IEEE, 2014.

\bibitem{ghaffari2014optimal1}
Mohsen Ghaffari, Bernhard Haeupler, and Madhu Sudan.
\newblock Optimal error rates for interactive coding i: Adaptivity and other
  settings.
\newblock In {\em Proceedings of the forty-sixth annual ACM symposium on Theory
  of computing}, pages 794--803. ACM, 2014.

\bibitem{haeupler2014interactive}
Bernhard Haeupler.
\newblock Interactive channel capacity revisited.
\newblock In {\em Foundations of Computer Science (FOCS), 2014 IEEE 55th Annual
  Symposium on}, pages 226--235. IEEE, 2014.

\bibitem{haeupler2011new}
Bernhard Haeupler, Barna Saha, and Aravind Srinivasan.
\newblock New constructive aspects of the lov{\'a}sz local lemma.
\newblock {\em Journal of the ACM (JACM)}, 58(6):28, 2011.

\bibitem{haeupler2017synchronization}
Bernhard Haeupler and Amirbehshad Shahrasbi.
\newblock Synchronization strings: codes for insertions and deletions
  approaching the singleton bound.
\newblock In {\em Proceedings of the 49th Annual ACM SIGACT Symposium on Theory
  of Computing}, pages 33--46. ACM, 2017.

\bibitem{HS17c}
Bernhard Haeupler and Amirbehshad Shahrasbi.
\newblock Synchronization strings: Explicit constructions, local decoding, and
  applications.
\newblock {\em arXiv preprint arXiv:1710.09795}, 2017.

\bibitem{haeupler2017synsimucode}
Bernhard Haeupler, Amirbehshad Shahrasbi, and Ellen Vitercik.
\newblock Synchronization strings: Channel simulations and interactive coding
  for insertions and deletions.
\newblock {\em arXiv preprint arXiv:1707.04233}, 2017.

\bibitem{kol2013interactive}
Gillat Kol and Ran Raz.
\newblock Interactive channel capacity.
\newblock In {\em Proceedings of the forty-fifth annual ACM symposium on Theory
  of computing}, pages 715--724. ACM, 2013.

\bibitem{mercier2010survey}
Hugues Mercier, Vijay~K Bhargava, and Vahid Tarokh.
\newblock A survey of error-correcting codes for channels with symbol
  synchronization errors.
\newblock {\em IEEE Communications Surveys \& Tutorials}, 12(1), 2010.

\bibitem{moenck1972fast}
Robert Moenck and Allan Borodin.
\newblock Fast modular transforms via division.
\newblock In {\em Switching and Automata Theory, 1972., IEEE Conference Record
  of 13th Annual Symposium on}, pages 90--96. IEEE, 1972.

\bibitem{moser2010constructive}
Robin~A Moser and G{\'a}bor Tardos.
\newblock A constructive proof of the general lov{\'a}sz local lemma.
\newblock {\em Journal of the ACM (JACM)}, 57(2):11, 2010.

\end{thebibliography}

\end{document}